\documentclass[11pt]{article}
\usepackage[dvips]{graphicx}
\usepackage{amssymb}
\usepackage{amsmath, amsthm}
\usepackage{setspace}

\newenvironment{packed_enum}{
\begin{enumerate}
  \setlength{\itemsep}{1pt}
  \setlength{\parskip}{0pt}
  \setlength{\parsep}{0pt}
}{\end{enumerate}}

\newtheorem{proposition}{Proposition}

\newtheorem{lemma}{Lemma}

\begin{document}

\title{\vspace{-1.5cm}Designing Incentive Schemes Based on Intervention:\\ The Case of Imperfect Monitoring}

\author{Jaeok Park and Mihaela van der Schaar\thanks{Electrical Engineering
Department, University of California, Los Angeles (UCLA). Email: \{jaeok, mihaela\}@ee.ucla.edu.}}

\date{}

\maketitle

\vspace{-1cm}

\begin{abstract}
We propose an incentive scheme based on intervention to sustain
cooperation among self-interested users. In the proposed scheme,
an intervention device collects imperfect signals about the
actions of the users for a test period, and then chooses
the level of intervention that degrades the performance of the network
for the remaining time period. We analyze the problems of designing an optimal intervention rule
given a test period and choosing an optimal length of the test period.
The intervention device can provide
the incentive for cooperation by exerting intervention following
signals that involve a high likelihood of deviation.
Increasing the length of the test period has two counteracting effects
on the performance: It improves the quality of signals, but at the same time it weakens the
incentive for cooperation due to increased delay.
\end{abstract}

\section{Introduction}

This paper studies incentive schemes to sustain cooperation among
self-interested users sharing a common network resource. When users utilize the network
resource considering their own self-interest, a problem known as the tragedy of the commons \cite{Hardin}
is likely to occur, yielding a suboptimal performance.
Different methods to overcome the problem have been
investigated in the literature. One method widely studied in economics
and engineering is pricing \cite{varian}. Pricing can induce efficient
use of network resources by internalizing negative externalities.
Although pricing has a solid theoretical foundation, implementing a pricing
scheme can be impractical or cumbersome in some cases. Let us consider a wireless
Internet service as an example. A service provider can limit access
to its network resources by charging an access fee. However, charging an access fee
requires a secure and reliable method to process payments, which creates
burden on both sides of users and service providers. There also arises
the issue of allocative fairness when a service provider charges for the Internet service.
In the presence of the income effect, pricing will bias the allocation
of network resources towards users with high incomes. Because the Internet
can play the role of an information equalizer, it has
been argued in a public policy debate that access to the Internet
should be provided as a public good by a public authority rather than
as a private good in a market \cite{Hallgren}.

Another method popular in game theory is to use repeated interaction \cite{mailath}.
Repeated interaction can encourage cooperative behavior by adjusting
future payoffs depending on current behavior. A repeated game strategy can
form a basis of an incentive scheme in which monitoring and punishment burden
is decentralized to users (see, for example, \cite{La}). However, implementing a repeated game strategy
requires repeated interaction among users, which may not be available.
For example, users interacting in a mobile network change frequently in nature.

In this paper, we use an alternative method based on intervention, which
was proposed in our previous work \cite{jpark}. In an incentive scheme based on
intervention, a network is augmented with an intervention device that is able to monitor the actions of users
and to take an action that affects the payoffs of users. In \cite{jpark}, we considered an ideal
scenario where the intervention device can observe the actions of users without errors
immediately after users choose their actions. In this paper, we consider
a more realistic scenario where the intervention device can obtain only imperfect
information about the actions of users and it takes time for the intervention device
to collect signals. Intervention directly affects the
network usage of users, unlike pricing which uses an outside instrument
to affect the payoffs of users. Thus, an incentive scheme based on intervention can provide
an effective and robust method to induce cooperation in that users
cannot avoid intervention as long as they use network resources.
Moreover, it does not require long-term relationship among users, which
makes it applicable to networks with a dynamically changing user population.

\section{Model and Problem Formulation}

We consider a communication channel shared by $N$ users. Time is divided
into slots of equal length, and in each slot a user can attempt to transmit
its packet or wait. If there is only one transmission attempt in a slot, the
packet is successfully transmitted. If there is more than one transmission
attempt in a slot, packets collide and no transmission is successful.
For simplicity, we assume that each user can choose one of two transmission probabilities
$p_l$ and $p_h$, where $p_l = 1/N < p_h <1$. Note that each user choosing $p_l$
maximizes the total throughput, defined as the average number of
successfully transmitted packets per time slot, assuming that all the users choose
the same transmission probability \cite{massey}.

We consider a period consisting of $T$ consecutive time slots, and analyze
interaction in the period without any consideration of past or future periods.
We assume that the number of users
and their transmission probabilities are fixed throughout a period.
Let $\mathcal{N} = \{1,\ldots,N\}$ be the set of the users.
The action space of user $i$ is denoted by $A_i = \{p_l, p_h\}$,
and the action of user $i$ is denoted by $a_i \in A_i$,
for all $i \in \mathcal{N}$. An action profile is represented by
a vector $\mathbf{a} = (a_1, \ldots, a_N) \in A \triangleq \prod_{i \in \mathcal{N}} A_i$.
The payoff of user $i$ is given by the number of its successfully transmitted
packets per time slot. Then the expected payoff of user $i$ is given
by the probability of its successful transmission, $a_i \prod_{j \in \mathcal{N} \setminus \{i\}} (1-a_j)$.
It is easy to see that the action $p_h$ is a dominant strategy for every user.
Hence, $(p_h,\ldots,p_h)$ is the unique Nash equilibrium, which
yields the lower total throughput than the symmetric social optimum $(p_l,\ldots,p_l)$.

In order to improve the inefficiency of Nash equilibrium, we introduce
an intervention device in the system. The intervention device is capable of
monitoring the actions of the users and interfering in the transmission of the users.
The intervention device can sense the channel
to learn whether the channel is idle (i.e., no user attempts to transmit its packet) or busy
(i.e., at least one user attempts to transmit its packet).
We consider a scenario where the intervention device collects signals
from sensing the channel for the first $t$ slots, where $1 \leq t \leq T$, and then chooses
its transmission probability, which can be interpreted as the intervention level.

Let $S = \{idle, busy\}$ be the set of all possible signals obtained in a slot. Then
the set of all possible signals that the intervention device can obtain for the $t$ slots
is $S^t$. The probability distribution of signals is independent
across slots, and when the users choose action profile $\mathbf{a}$, the probability
of obtaining an idle signal in a slot is given by $q(\mathbf{a}) \triangleq \prod_{i \in \mathcal{N}} (1-a_i)$.
After obtaining $t$ signals, the intervention device
chooses a transmission probability in $[0,1]$, which remains fixed
until the end of the period. We use subscript 0 for the intervention device.
The action space of the intervention device is denoted by $A_0 = [0,1]$,
and its action is denoted by $a_0 \in A_0$. The decision rule of the intervention
device, called the intervention rule, can be represented by a function
$f:S^t \rightarrow A_0$. Since the transmission probabilities of the users do not
change in a period, there is no gain for the intervention device to distinguish
signals from different slots. Hence, we focus on the class of intervention rules
that use only the number of idle signals, which can be represented by
$f:\{0,1,\ldots,t\} \rightarrow A_0$. The probability that $k$ idle signals arise
out of $t$ signals when the users choose action profile $\mathbf{a}$ is
$\binom{t}{k} q(\mathbf{a})^k (1-q(\mathbf{a}))^{t-k}$, for $k=0,1,\ldots,t$.
Note that monitoring is imperfect in the sense that the intervention device cannot
observe the action profile of the users but obtains only imperfect information about
the action profile.

The sequence of events in a period can be listed as follows.
\begin{packed_enum}
\item At the beginning of the period, the users choose their transmission
probabilities $\mathbf{a} \in A$, which are used from slot 1 to slot $T$, knowing
the intervention rule $f$ adopted by the intervention device.
\item The intervention device collects signals from slot 1 to slot $t$.
\item The intervention device intervenes using the transmission probability
prescribed by the intervention rule $f$ from slot $t+1$ to slot $T$.
\end{packed_enum}

The payoff of user $i$ when the users choose action profile $\mathbf{a}$
and the intervention device chooses action $a_0$ is given by
\begin{align}
u_i(a_0, \mathbf{a}) &= \frac{t}{T} a_i \prod_{j \in \mathcal{N} \setminus \{i\}} (1-a_j)
+ \frac{T-t}{T} a_i (1-a_0) \prod_{j \in \mathcal{N} \setminus \{i\}} (1-a_j)\\
&= \left( 1 - \frac{T-t}{T} a_0 \right) a_i \prod_{j \in \mathcal{N} \setminus \{i\}} (1-a_j).
\end{align}
The action profile of the users influences the probability distribution of signals,
which in turn affects the action of the intervention device. The expected
payoff of user $i$ when the users choose action profile $\mathbf{a}$
and the intervention device uses intervention rule $f$ can be expressed as
\begin{align}
v_i(\mathbf{a};f) &= \sum_{k=0}^t \binom{t}{k} q(\mathbf{a})^k (1-q(\mathbf{a}))^{t-k} u_i(f(k), \mathbf{a})\\
&= \left[ 1 - \frac{T-t}{T} \sum_{k=0}^t \binom{t}{k} q(\mathbf{a})^k (1-q(\mathbf{a}))^{t-k} f(k) \right]
a_i \prod_{j \in \mathcal{N} \setminus \{i\}} (1-a_j).
\end{align}
Note that $\sum_{k=0}^t \binom{t}{k} q(\mathbf{a})^k (1-q(\mathbf{a}))^{t-k} f(k)$
can be interpreted as the expected transmission probability of the
intervention device, while $(T-t)/T$ is the weight on the slots in which
the action of the intervention device affects the users.

For notation, let us define
\begin{align}
\lambda(k;t) &= \binom{t}{k} [(1-p_l)^N]^k [1-(1-p_l)^N]^{t-k},\\
\mu(k;t) &= \binom{t}{k} [(1-p_l)^{N-1}(1-p_h)]^k [1-(1-p_l)^{N-1}(1-p_h)]^{t-k},
\end{align}
for $k=0,1,\ldots,t$, and let $\tau_c = p_l(1-p_l)^{N-1}$ and $\tau_d = p_h(1-p_l)^{N-1}$.
$\lambda(k;t)$ is the probability of $k$ idle signals arising out of
$t$ signals when every user cooperates (i.e., chooses $p_l$), while $\mu(k;t)$ is that when
exactly one user defects (i.e., chooses $p_h$). $\tau_c$ is the cooperation throughput
that each user obtains when all the users choose $p_l$, while $\tau_d$
is the defection throughput that a user obtains when it deviates
to $p_h$ unilaterally. Note that an idle signal is more likely to occur when
every user cooperates than when some user defects. Also, note that $\tau_d > \tau_c$,
which reflects the positive gain from defection when there is no intervention.

Suppose that there is a network manager who determines the intervention
rule used by the intervention device. The objective of the manager is to maximize
the sum of the payoffs (i.e., total throughput) while sustaining cooperation
among the users. The cooperation payoff is given by
\begin{align}
\left[ 1 - \frac{T-t}{T} \sum_{k=0}^t \lambda(k;t) f(k) \right] \tau_c,
\end{align}
while the defection payoff is
\begin{align}
\left[ 1 - \frac{T-t}{T} \sum_{k=0}^t \mu(k;t) f(k) \right] \tau_d.
\end{align}
Hence, the incentive constraint for the users to cooperate can be
written as
\begin{align}
\left[ 1 - \frac{T-t}{T} \sum_{k=0}^t \lambda(k;t) f(k) \right] \tau_c
\geq \left[ 1 - \frac{T-t}{T} \sum_{k=0}^t \mu(k;t) f(k) \right] \tau_d,
\end{align}
and the problem of designing an intervention rule can be expressed as
\begin{align}
&\max_{f} N \left[ 1 - \frac{T-t}{T} \sum_{k=0}^t \lambda(k;t) f(k) \right] \tau_c \label{eq:pdp1}\\
&\text{subject to } \left[ 1 - \frac{T-t}{T} \sum_{k=0}^t \lambda(k;t) f(k) \right] \tau_c
\geq \left[ 1 - \frac{T-t}{T} \sum_{k=0}^t \mu(k;t) f(k) \right] \tau_d, \label{eq:pdp2}\\
&\qquad \qquad \quad 0 \leq f(k) \leq 1 \text{ for all $k = 0, \ldots, t$}.\label{eq:pdp3}
\end{align}

\section{Analysis of the Design Problem}

The design problem \eqref{eq:pdp1}--\eqref{eq:pdp3} can be rewritten as
a linear programming (LP) problem,
\begin{align}
&\min_{f} \sum_{k=0}^t \lambda(k;t) f(k) \label{eq:lp1}\\
&\text{subject to } \frac{T-t}{T} \sum_{k=0}^t [\tau_d \mu(k;t) - \tau_c \lambda(k;t)] f(k)
\geq \tau_d - \tau_c \label{eq:lp2}\\
&\qquad \qquad \quad 0 \leq f(k) \leq 1 \text{ for all $k = 0, \ldots, t$}.\label{eq:lp3}
\end{align}
The LP problem \eqref{eq:lp1}--\eqref{eq:lp3} is to minimize the
expected transmission probability of the intervention device while
satisfying the incentive constraint and the probability constraints.
Exerting intervention is necessary to punish a deviation, but
at the same time intervention incurs efficiency loss under imperfect
monitoring. Therefore, the manager wants to use the minimum possible intervention
level while providing the incentive for cooperation. The left-hand side of
the incentive constraint \eqref{eq:lp2} is the expected loss from
deviation due to the change in the probability distribution of signals
induced by deviation, while the right-hand side is the gain from
deviation.

\begin{lemma}
Suppose that an optimal solution to the LP problem \eqref{eq:lp1}--\eqref{eq:lp3}
exists. Then the incentive constraint \eqref{eq:lp2} is satisfied
with equality at the optimal solution.
\end{lemma}

\begin{proof}
Let $f^*$ be an optimal solution. Suppose that
$[(T-t)/T] \sum_{k=0}^t [\tau_d \mu(k;t) - \tau_c \lambda(k;t)] f^*(k)
> \tau_d - \tau_c$. Since $\tau_d > \tau_c$, there exists $k'$ such
that $\tau_d \mu(k';t) - \tau_c \lambda(k';t) > 0$ and $f^*(k') > 0$.
Then we can reduce $f^*(k')$ while satisfying the incentive constraint
and the probability constraint for $k'$, which decreases the objective
value since $\lambda(k;t) > 0$ for all $k$. This contradicts the optimality
of $f^*$.
\end{proof}

\noindent
Lemma 1 validates the intuition that the manager wants to use
a punishment just enough to prevent deviation. The following
proposition provides a necessary and sufficient condition
for the LP problem to have a feasible solution, and the structure
of an optimal solution.

\begin{proposition}
Let $k_0 = \max \{k : \tau_d \mu(k;t) - \tau_c \lambda(k;t) > 0 \}$.
Then the LP problem has a feasible solution if and only if
\begin{align} \label{eq:feas}
\frac{T-t}{T} \sum_{k \leq k_0} [\tau_d \mu(k;t) - \tau_c \lambda(k;t)]
\geq \tau_d - \tau_c.
\end{align}
Moreover, if the LP problem has a feasible solution, then there exists
a unique optimal solution $f^*$ described by
\begin{eqnarray}
f^*(k) = \left\{ \begin{array}{ll}
1 & \textrm{if $k<\bar{k}$,}\\
\frac{1}{\tau_d \mu(\bar{k};t) - \tau_c \lambda(\bar{k};t)} \left[
\frac{T}{T-t} (\tau_d - \tau_c) - \sum_{k=0}^{\bar{k}-1} [\tau_d \mu(k;t) - \tau_c \lambda(k;t)] \right] & \textrm{if $k=\bar{k}$,}\\
0 & \textrm{if $k>\bar{k}$,}
\end{array} \right.
\end{eqnarray}
where
\begin{align}
\bar{k} = \min \left\{k': \frac{T-t}{T} \sum_{k \leq k'} [\tau_d \mu(k;t) - \tau_c \lambda(k;t)]
\geq \tau_d - \tau_c \right\}.
\end{align}
\end{proposition}

\begin{proof}
Define the likelihood ratio of signal $k$ (i.e., $k$ idle signals out of $t$
signals) by
\begin{align}
L(k;t) = \frac{\mu(k;t)}{\lambda(k;t)} = \left( \frac{1-p_h}{1-p_l} \right)^k
\left( \frac{1-(1-p_l)^{N-1}(1-p_h)}{1-(1-p_l)^N} \right)^{t-k}.
\end{align}
It is easy to see that $L(0;t) > 1$, $L(t;t) < 1$, and $L(k;t)$ is monotonically
decreasing in $k$. Note that $\tau_d \mu(k;t) - \tau_c \lambda(k;t) > 0$
if and only if $L(k;t) > p_l/p_h$. Hence, $k_0$ is well-defined, and
$\tau_d \mu(k;t) - \tau_c \lambda(k;t) > 0$ if and only if $k \leq k_0$.
If \eqref{eq:feas} is satisfied, then $\tilde{f}$ defined by $\tilde{f}(k) = 1$ for all $k \leq k_0$
and $\tilde{f}(k) = 0$ for all $k > k_0$ is a feasible solution. To prove the converse,
suppose that a feasible solution, say $f$, exists. Then we have
\begin{align}
\frac{T-t}{T} \sum_{k \leq k_0} [\tau_d \mu(k;t) - \tau_c \lambda(k;t)] \geq
\frac{T-t}{T} \sum_{k=0}^t [\tau_d \mu(k;t) - \tau_c \lambda(k;t)] f(k)
\end{align}
and
\begin{align}
\frac{T-t}{T} \sum_{k=0}^t [\tau_d \mu(k;t) - \tau_c \lambda(k;t)] f(k)
\geq \tau_d - \tau_c,
\end{align}
and combining the two yields \eqref{eq:feas}.

To prove the second result, suppose that the LP problem has
a feasible solution. Then there exists a feasible solution, say $f$, that satisfies
the incentive constraint with equality. Define the likelihood ratio of $f$ by
\begin{align}
l(f) = \frac{\sum_k \mu(k;t) f(k)}{\sum_k \lambda(k;t) f(k)}.
\end{align}
Then the objective value in \eqref{eq:lp1} at $f$ can be expressed as
\begin{align}
\frac{T}{T-t} \frac{\tau_d - \tau_c}{\tau_d l(f) - \tau_c}.
\end{align}
Hence, the objective value decreases as $f$ has a larger likelihood ratio.
To optimize the objective value, $f$ should put the probabilities on the signals
starting from signal 0 to signal 1, and so on, until the incentive constraint
is satisfied with equality. Thus, we obtain $\bar{k}$, where $0 \leq \bar{k} \leq
k_0$, that is associated with the unique optimal solution.
\end{proof}

Since a smaller number of idle signals gives a higher likelihood ratio, an
intervention rule yields a smaller efficiency loss when intervention is
exerted following a smaller number of idle signals. Put differently,
signal $k$ provides a stronger indication of defection as $k$ is smaller.
However, using only signal 0 may not be sufficient to provide
the incentive for cooperation, in which case other signals need to be
used as well. Using signal $k$ with $k \leq k_0$ contributes to provide
the incentive for cooperation, although the ``quality'' of the signal
decreases as $k$ increases. Hence, it is optimal for the manager to
use signals with small $k$ primarily, which yields a threshold $\bar{k}$.

So far we have analyzed the problem of designing an intervention rule when
the total number of signals, $t$, is fixed. Now we consider a scenario
where the manager can choose $t$ as well as an intervention rule.
In this scenario, there are two counteracting effects of increasing $t$.
First, note that the objective value in \eqref{eq:pdp1} can be expressed as
\begin{align}
N \left[ 1 - \frac{\tau_d - \tau_c}{\tau_d l(f) - \tau_c} \right] \tau_c,
\end{align}
which shows that increasing $t$ affects the objective value only through $f$.
Since $L(k;t)$ is increasing in $t$, we can achieve a larger likelihood ratio
$l(f)$ with larger $t$. In other words, as the intervention device collects
more signals, the information becomes more accurate (quality effect). On the other hand,
increasing $t$ decreases the weight given on the slots in which intervention
is applied, which makes the incentive constraint harder to satisfy (delay
effect).

Let $\tau^*(t)$ be the optimal value of the design problem \eqref{eq:pdp1}--\eqref{eq:pdp3},
where we set $\tau^*(t) = N p_h (1-p_h)^{N-1}$ if there is no feasible
solution with $t$. The problem of finding an optimal number of signals
can be written as $\max_{t \in \{1,\ldots,T\}} \tau^*(t)$.
In general, $\tau^*(t)$ is a non-monotonic function of $t$, and we
provide a numerical example to illustrate the result. We consider
system parameters $N=5$, $p_l = 1/N = 0.2$, $p_h = 0.8$, and $T = 100$. Then we have
$\tau_c = 0.08$ and $\tau_d = 0.33$. The numerical results show that the LP problem
is infeasible for $t=1$ and $t \geq 21$. With $t=1$, there is not
sufficient information based on which intervention can provide the incentive
for cooperation. With $t \geq 21$, the delay effect is dominant, which prevents
the incentive constraint to be satisfied. Figure~\ref{fig:dura} plots $\tau^*(t)$
for $t = 2, \ldots, 20$. We can see that $\tau^*(t)$ is non-monotonic
while reaching the maximum at $t = 18$ with $\tau^*(18) = 0.37$. In the plot,
the dotted line represents the total throughput at $(p_l,\ldots,p_l)$, $N \tau_c$.
The difference between $\tau^*(t)$ and $N \tau_c$ can be interpreted as
the efficiency loss due to imperfect monitoring.\footnote{If the
intervention device can observe the actions of the users
immediately, it can use the threat of transmitting with probability 1
when a deviation is detected to sustain cooperation without incurring an efficiency loss.}
Lastly, we note that $\bar{k}$ in Proposition 1 is non-decreasing in $t$,
with $\bar{k}=1$ for $t=2,\ldots,7$, $\bar{k}=2$ for $t=8,\ldots,13$,
$\bar{k}=3$ for $t=14,\ldots,18$, and $\bar{k}=4$ for $t=19,20$.

\begin{figure}
\centering
\includegraphics[width=0.6\textwidth]{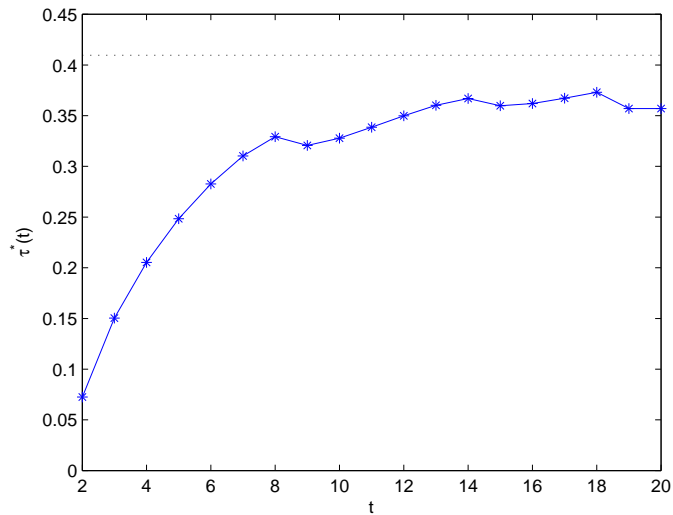}
\caption{Plot of $\tau^*(t)$ for $t = 2, \ldots, 20$.}
\label{fig:dura}
\end{figure}

\section{Conclusion}

We have studied the problem of designing incentive schemes based on
the idea of intervention to sustain cooperation among users sharing
network resources in the case of imperfect monitoring. We have used a
simple model to present the main ideas and results without too many complications.
Our model can be extended in several directions, among which we mention two.
First, users can use more complicated decision rules than the one
choosing one of two transmission probabilities. The action space for a user
can be relaxed to $[0,1]$ instead of $\{p_l,p_h\}$. Also, users can have
an ability to monitor the actions of other users and the intervention device.
In such a scenario, we can  study intervention rules to sustain a cooperative decision rule, where
a decision rule for a user is a mapping from its information set to
its action space. Second, the set of signals that the intervention
device can obtain in a slot can be generalized. For example, a signal
can be ternary (idle, success, collision) or the number of users that
attempted to transmit. We can investigate how optimal intervention
rules and their performance change as the intervention device
obtains finer information about the actions of users.
Finally, we conclude with a remark that incentive schemes based on intervention
can be applied to a wide range of networks where
cooperative behavior should be encouraged. Potential applications
include communication networks (power control, congestion control,
and medium access control) and peer-to-peer networks.

\small
\singlespacing


\begin{thebibliography}{99}

\bibitem{Hardin} G. Hardin, ``The tragedy of the commons,''
\emph{Science}, vol. 162, no. 3859, pp. 1243--1248, Dec. 1968.

\bibitem{varian} J. K. MacKie-Mason and H. R. Varian, ``Pricing congestible
network resources,'' \emph{IEEE J. Sel. Areas Commun.}, vol. 13, no. 7, pp.
1141--1149, Sep. 1995.

\bibitem{Hallgren} M. M. Hallgren and A. K. McAdams, ``The economic efficiency
of Internet public goods,'' in \emph{Internet Economics}, L. W. McKnight and J.
P. Bailey, Eds. Cambridge, MA: MIT Press, 1997, pp. 455--478.

\bibitem{mailath} G. Mailath and L. Samuelson, \emph{Repeated Games
and Reputations: Long-run Relationships}. Oxford, U.K.: Oxford Univ.
Press, 2006.

\bibitem{La} R. J. La and V. Anantharam, ``Optimal routing control: repeated
game approach,'' \emph{IEEE Trans. Autom. Control}, vol.47, no.3, pp.437--450, Mar. 2002.

\bibitem{jpark} J. Park and M. van der Schaar, ``Stackelberg contention games
in multiuser networks,'' {\it EURASIP J. Advances Signal Process.}, vol. 2009,
Article ID 305978, 15 pages, 2009.

\bibitem{massey} J. L. Massey and P. Mathys, ``The collision channel without feedback,''
\emph{IEEE Trans. Inf. Theory}, vol. 31, no. 2, pp. 192--204, Mar. 1985.

\end{thebibliography}
\end{document}